\documentclass[journal,twocolumn]{IEEEtran}
\usepackage{cite}
\usepackage[T1]{fontenc} % I used xelatex
\usepackage{amsmath,amssymb,amsfonts,amsthm}
\usepackage{algorithm}
\usepackage{algorithmic}
\usepackage{graphicx}
\graphicspath{{./visiofigures/}{./SimFigures/}{./BioFigures/}}
\usepackage{textcomp, epstopdf}
\usepackage{xcolor}
\usepackage{booktabs,multirow}
\usepackage{tabularx}
\usepackage[inline]{enumitem}
\usepackage{cite}
\usepackage{bigints, color}
\usepackage[unicode=true,
 bookmarks=false,bookmarksnumbered=true,bookmarksopen=false,bookmarksopenlevel=1,
 breaklinks=false,pdfborder={0 0 0},pdfborderstyle={},backref=false,colorlinks=false]
 {hyperref}
\hypersetup{pdftitle={Title},
 pdfauthor={Admin},
 pdfpagelayout=OneColumn, pdfnewwindow=true, pdfstartview=XYZ, plainpages=false}
 
\makeatletter

\DeclareMathOperator*{\argmax}{arg\,max\ }

\newtheorem{theorem}{Theorem}
\theoremstyle{remark}

\makeatother
\allowdisplaybreaks
\begin{document}
\title{Joint Channel Estimation and Mixed-ADCs Allocation for Massive MIMO via Deep Learning}
\author{Liangyuan~Xu, Feifei~Gao, Ting~Zhou, Shaodan~Ma, and~Wei~Zhang
% \thanks{Manuscript received January 2, 2020; revised July 21, 2020; accepted October 2, 2020. The editor coordinating the review of this manuscript and approving it for publication was Dr. David Matolak.  
% This work was supported in part by National Key Research and Development Program of China (2018AAA0102401), by the National Natural Science Foundation of China under Grant \{61831013, 
% 61771274,61531011\}, by Beijing Municipal Natural Science Foundation under Grant (4182030, L182042).
% The work of W.~Zhang was supported in part by Shenzhen Science and Innovation Fund, under Grant JCYJ20180507182451820 and JCYJ20170412104656685.
% The work of S.~Ma was supported in part by the Science and Technology Development Fund, Macau SAR (File no. 0036/2019/A1 and File no. SKL-IOTSC2018-2020), and by the Research Committee of University of Macau under Grant MYRG2018-00156-FST.
% Parts of this paper were presented at the IEEE Global Communications Conference 2019 \cite{XLYconf}. \emph{(Corresponding author: Feifei Gao.)}}
\thanks{L.~Xu and F.~Gao  are with the Department of Automation, Tsinghua University, Beijing 100084, China, and also with Beijing National Research Center for Information Science and Technology (BNRist), Beijing 100084, China (email: \protect\href{mailto:xly18@mails.tsinghua.edu.cn}{xly18@mails.tsinghua.edu.cn}; \protect\href{mailto:feifeigao@ieee.org}{feifeigao@ieee.org}).}
\thanks{T. Zhou is with the Shanghai Frontier Innovation Research Institute, Chinese Academy of Sciences, Shanghai 201210, P.R. China (email: \protect\href{mailto:zhouting@sari.ac.cn}{zhouting@sari.ac.cn}).}
\thanks{S.~Ma is with the State Key Laboratory of Internet of Things for Smart City and the Department of Electrical and Computer Engineering, University of Macau, Macao S.A.R. 999078, China (email: \protect\href{mailto:shaodanma@um.edu.mo}{shaodanma@um.edu.mo}).}
\thanks{W.~Zhang is with the School of Electrical Engineering and Telecommunications, The University of New South Wales, Sydney, NSW 2052, Australia (email: \protect\href{mailto:w.zhang@unsw.edu.au}{w.zhang@unsw.edu.au}).}
}
\maketitle
\begin{abstract}
Millimeter wave (mmWave) multi-user  massive multi-input multi-output (MIMO) is a promising technique for the next generation communication systems. 
However, the hardware cost and power consumption grow significantly as the number of radio frequency (RF) components increases, which hampers the deployment of practical massive MIMO systems.
To address this issue and further facilitate the commercialization of massive MIMO, mixed analog-to-digital converters (ADCs) architecture has been considered, where parts of conventionally assumed full-resolution ADCs are replaced by one-bit ADCs.
In this paper, we first propose a deep learning-based (DL) joint pilot design and channel estimation method for mixed-ADCs mmWave massive MIMO.
Specifically,  we devise a pilot design neural network whose weights directly represent the optimized pilots, and develop a Runge-Kutta model-driven densely connected network as the channel estimator. 
Instead of  randomly assigning the mixed-ADCs, we then design a novel antenna selection network for mixed-ADCs allocation to further improve the channel estimation accuracy.
Moreover, we adopt an autoencoder-inspired end-to-end architecture to jointly optimize the pilot design, channel estimation and mixed-ADCs allocation networks. 
Simulation results show that the proposed DL-based methods have advantages over the traditional channel estimators as well as the state-of-the-art networks.
\end{abstract}

\begin{IEEEkeywords}
MmWave massive MIMO, deep learning, channel estimation, pilot design, antenna selection, one-bit quantization, mixed-ADC.
\end{IEEEkeywords}

\section{Introduction}
\IEEEPARstart{M}{illimeter} wave (mmWave) massive multi-input multi-output (MIMO), one of the promising technology for both 5G and the next generation mobile communication systems, has attracted tremendous attention for its several appealing advantages\cite{five5G}.  By deploying a massive array with hundreds or thousands of antennas at the base station (BS), massive MIMO can improve spectral efficiency and enhance cellular coverage\cite{benefits}.
Meanwhile, with the 30-300GHz mmWave frequency band, the array inter-element spacing can be reduced significantly, which allows to pack a large scale of antenna array into a small area at the BS. 
Furthermore, with highly directional beamforming, the massive antenna array can concentrate energy  on  the sparse mmWave channel paths to facilitate spatial multiplex and increase energy efficiency \cite{mmWave}. 

Although mmWave massive MIMO has many irresistible advantages, one major obstacle in its commercialization is the financial cost and power consumption  that grow exponentially as the bit resolution of analog-to-digital converters (ADCs) increases \cite{waldenADC}. An alternative solution for this problem is to replace the high-resolution ADCs with economical low-resolution ones, e.g., single quantization bit ADCs \cite{ADCsurveyDll}.
However,  the traditional optimal pilot sequence design, channel estimation and data detection techniques are devised for finely quantized data. Hence, the nonlinear quantization errors introduced by low-resolution ADCs will make the traditional methods no longer eligible, which motivates the studies of new algorithms for mmWave massive MIMO with low-resolution ADCs \cite{HardwareSurvey}.

% In the literature, there have been numerous results for massive MIMO systems with low-resolution ADCs, e.g., channel estimation \cite{JDC,nML,MJH,blind,FJesti,WenDoA,AR,XLYonebit}, data detection \cite{BayesADC,RLDetect,YXDetect,SVMDetect}, beamforming \cite{beam,YJbeam,JCbeam} and performance analysis \cite{FLanaly,CManaly,SJanaly,LYZanaly,DPHanaly,XWanaly,XLY}.
There are numerous design for massive MIMO systems with low-resolution ADCs, e.g., channel estimation \cite{WenDoA,AR,XLYonebit,AhmedPilot}, data detection \cite{RLDetect,YXDetect,SVMDetect}, beamforming \cite{beam,YJbeam,JCbeam} and performance analysis \cite{CManaly,DPHanaly,XLY}.
To leverage the potential gains offered by the large scale antenna array, accurate channel state information (CSI) should be obtained at the BS. However, it is challenging to acquire CSI especially when one-bit ADCs are utilized, since the the amplitude information of the received signal is lost during quantization.
In \cite{WenDoA}, the authors have addressed channel estimation of massive MIMO with one-bit ADCs by formulating the problem as the atomic norm optimization.
In \cite{AR}, the authors have proposed an amplitude retrieval (AR) algorithm for channel estimation with one-bit ADCs. This algorithm performs channel estimation by completing the lost amplitudes and recovery the direction of 
arrival (DOA). 
The authors of \cite{XLYonebit} have considered a compressed sensing based channel estimation method and a gridless angular domain sparse parameters estimation algorithm to further improve the one-bit channel estimation accuracy. 
Nevertheless, the methods therein have ignored the pilot design problem and adopted general pilot sequences, which would cause high training overhead. The work \cite{AhmedPilot} has shown that massive MIMO with one-bit ADCs can achieve the same channel estimation performance with fewer pilots under the assumption of single-path channel model. However, it is hard to solve the joint pilot design and channel estimation problem  due to the non-convex constraints introduced by the low-resolution ADCs.

Recently, numerous works have shown the successful applications of deep learning (DL) in communication systems, which includes channel estimation \cite{LYOFDM,YYWDL,1bOFDM,cGan}, pilot design \cite{PilotSIC,GZpilot,DNNCEPD,PilotPattern}, CSI feedback \cite{CSINet,XWCSI,GJJ,YHY}, data detection \cite{DeepMIMODetect,JCESD,LJYDetect,LearnSearch}, \emph{et al.}
DL algorithms exhibits superior performance in tackling the one-bit channel estimation problem that the traditional methods are unable to handle well. Specifically, the authors in \cite{1bOFDM} have devised a novel DL-based architectures for the one-bit quantized orthogonal frequency division multiplexing (OFDM) receiver to deal with the channel estimation and data detection problem. The work in \cite{DNNCEPD} has developed a deep neural network-based (DNN)  channel estimator and training signal design for low-resolution quantized MIMO systems. However the neural network is a simple fully connected one and hence exhibit poor performance. In \cite{cGan}, a conditional generative adversarial network (cGAN) has been developed to estimate the channel matrix from the one-bit quantized received signals, which outperforms traditional channel estimators and simple neural networks, e.g., naive convolutional neural network (CNN). Though DL-based approaches have shown great potential in coarsely quantized massive MIMO systems, DL-based joint pilot design and channel estimation still faces challenges, e.g., simple network architecture and poor performance.

To further improve the performance of one-bit quantized massive MIMO systems,  a mixed-ADCs architecture has been proposed in \cite{mixed}, where parts of one-bit ADCs are replaced by full-resolution ADCs. The authors in \cite{mixedXW} have derived an approximate tractable expression for the uplink achievable rate of  massive MIMO with mixed-ADCs, and the results show the good trade-offs between performance and hardware cost offered by the mixed-ADCs architecture. A data detector for mixed-ADCs massive MIMO has been devised in \cite{mixedWCK} through probabilistic Bayesian inference. In \cite{mixedLY}, fully connected neural networks have been applied  to estimate the uplink channel matrix of mixed-ADCs massive MIMO. Furthermore, the authors of \cite{mixedYXH} have modified the DL-based approach in \cite{mixedLY} and developed a two-stage channel estimation method to further reduce channel estimation error. However, an interesting question has not been answered yet: how to optimize the allocation of the mixed-ADCs to improve the performance?

In this paper, we design a DL-based approach for joint channel estimation and pilot design as well as the allocation of the mixed-ADCs in mmWave massive MIMO system.
The contributions of this work are summarized as follows:
\begin{itemize}
\item {\bf Channel Estimation:}
For the channel estimation neural network, we develop densely connected blocks to deal with the mixed-precision quantized signals.
Unlike the naive CNN, each block in the dense connection is a deep-unfolding of the 3rd order Runge-Kutta method, which can alleviate the vanishing-gradient problem and exhibit superior  performances over the state-of-the-art channel estimation algorithms. 

\item {\bf Joint Pilot Design and Channel Estimation:}
We devise a pilot design subnet to  cooperate with the aforementioned channel estimation subnetwork. By developing an autoencoder-inspired end-to-end architecture and jointly training the subnets, we can obtain  the optimized pilot and channel estimator. The non-differentiable problem introduced by one-bit quantization is addressed by using a differentiable function as an approximation.

\item {\bf Optimization of The  Mixed-ADCs Allocation:}
By applying softmax and top-$K$ functions in  the output layer of the neural network, we  generate  two  selection vectors to choose the antennas for the mixed-ADCs allocation. Then,  we develop a novel training method to solve the non-differentiable problem caused by the top-$K$ function. By taking the $K$-hot encoding constraints into the loss function and training the network with the proposed method, we then optimize the  allocation of the mixed-ADCs.  
\end{itemize}

The remainder of this paper is organized as follows. In Section \ref{sec:UL CE}, we investigate the channel estimation problem with mixed-ADCs and propose the Runge-Kutta model-driven densely connected network. Section \ref{sec:Mixed} presents the antenna selection network and the proposed training method. In Section \ref{sec:Pilot}, the pilot design network, the autoencoder-inspired end-to-end architecture and the solution for the non-differentiable problem are illustrated. Numerical results are provided in Section \ref{sec:sim}, and the performances of the proposed methods are compared with benchmark algorithms. Conclusions are made in Section \ref{sec:con}.

\noindent\textbf{Notation:} Uppercase boldface $\mathbf{X}$, lowercase boldface $\mathbf{x}$ and lowercase non-bold $x$ denote matrices, vectors and scalars, respectively. Superscripts $(\cdot)^T$,  $(\cdot)^H$, $(\cdot)^{-1}$ and $(\cdot)^\dagger$ represent transpose, Hermitian transpose, matrix inverse and pseudo inverse, respectively. We use  $\|\cdot\|_F$ and $\|\cdot\|_p$ for Frobenius norm and $\ell_p$-norm, respectively.  The $m$th element of a vector $\mathbf{a}$ is $[a]_m$, and the $(m,n)$th element of the matrix $\mathbf{A}$ is $[\mathbf{A}]_{m,n}$. Symbols $\odot$ denote the element-wise product of two matrices. The cardinality of set $\mathcal{A}$ is given by $|\mathcal{A}|$. Symbol $\emptyset$ denotes empty set. The real and imaginary parts of complex numbers are denoted by $\Re(\cdot)$ and $\Im(\cdot)$, respectively. $\jmath = \sqrt{-1}$. ${\mathop{\rm sign}\nolimits} \left( \cdot \right)$ is the sign function. The cumulative distribution function (CDF) of the standard normal distribution is denoted by $\Phi(\cdot)$.

\section{Channel Estimation} \label{sec:UL CE}
\subsection{System Model}
\begin{figure}[!tpb]
\centering
\includegraphics[width=0.46\textwidth]{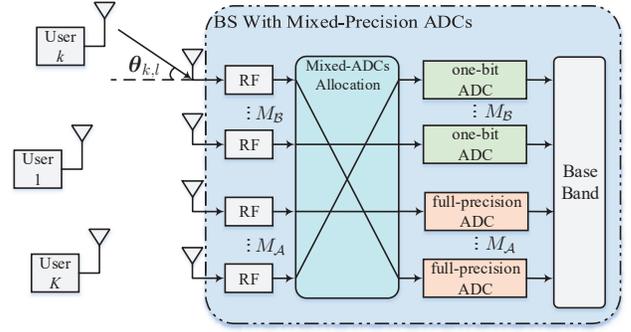}
\caption{Multi-user mmWave massive MIMO system with mixed-precision ADCs deployed at the BS.}
\label{fig:systdiagram}
\end{figure}
Consider an uplink multi-user mmWave massive MIMO system where $K$ single antenna users are served by a BS. The BS is equipped with a uniform linear array (ULA) of $M$ antennas with half-wavelength inter-element spacing. The uplink channel between the $k$th user and the BS is 
\begin{equation}\label{eq:ULkch}
\mathbf{h}_{k} = \sum\limits_{l=1}^{L_{k}}{\alpha_{k, l} \mathbf{a}\left(\theta_{k, l}\right)},
\end{equation}
where $L_{k}$ is the number of dominant paths between user-$k$ and the BS, $\alpha_{k, l}$ and $\theta_{k, l}$ denote the gain and DOA of the $l$th path respectively, and $\mathbf{a}\left(\theta_{k, l}\right) = \left[ {1,{e^{ - \jmath \pi \sin \left( {\theta _{k,l}} \right)}}, \cdots ,{e^{ - \jmath \left( M - 1 \right)\pi \sin \left( {\theta _{k,l}} \right)}}} \right]^T$ is the steering vector. 

We assume that mixed-precision ADCs (full-resolution and one-bit ADCs) are  deployed at the BS, as shown in Fig. \ref{fig:systdiagram}. During pilot symbol transmission, each user transmits a length-$N_{p}$ pilot sequence. The received signal $\mathbf{Y}\in {\mathbb C}^{M\times N_{p}}$ of the BS can be written in matrix form as 
\begin{equation}\label{eq:RxVector}
\mathbf{Y} = \mathcal{Q}\big(\mathbf{H} \mathbf{P} + \mathbf{W}\big),
\end{equation}
where ${\mathcal Q}(\cdot)$ is the element-wise quantization function, $\mathbf{H} =[\mathbf{h}_{1},\cdots,\mathbf{h}_{K}]\in {\mathbb C}^{M\times K}$ denotes the channel matrix, $\mathbf{P}\in {\mathbb C}^{K\times N_{p}}$ is the pilot matrix transmitted by $k$ users, and $\mathbf{W}\in {\mathbb C}^{M\times N_{p}}$ is the additive Gaussian noise with zero mean and variance $\sigma^2$. With mixed-precision ADCs, ${\mathcal Q}(\cdot)$ has different forms for full-resolution and one-bit ADCs respectively. If the RF chain is followed by full-resolution ADC, then ${\mathcal Q}(\cdot)$ is an identity function, and ${\mathcal Q}(\cdot)=\mathrm{sign}(\Re(\cdot)) + \jmath\mathrm{sign}(\Im(\cdot))$ holds for the case of one-bit ADC. Hence, \eqref{eq:RxVector} can be rewritten as 
\begin{equation}\label{eq:MixedQ}
[\mathbf{Y}]_{m,n} = \left\{\begin{array}{ll}
{[\mathbf{V}]_{m,n},} & m \in \mathcal{A} \\
{\mathrm{sign}\left(\Re\left([\mathbf{V}]_{m,n}\right)\right) + \jmath\mathrm{sign}(\Im([\mathbf{V}]_{m,n}))}, & m \in \mathcal{B}
\end{array}\right. \notag
\end{equation}
where $\mathbf{V} = \mathbf{H} \mathbf{P} + \mathbf{W}$, $\mathcal{A}$ and $\mathcal{B}$ are the index sets of RF chains followed by full-resolution and one-bit ADCs, respectively. Moreover, there are $\mathcal{A} \cup \mathcal{B}=\{1,2, \ldots, M\}$, $\mathcal{A} \cap \mathcal{B}=\emptyset$, $|\mathcal{A}|=M_{\mathcal{A}}$ and $|\mathcal{B}|=M_{\mathcal{B}}=M-M_{\mathcal{A}}$.

\subsection{The Proposed Channel Estimation Network}\label{subsec:CENet}
Assume that the pilot matrix $\mathbf{P}$ as well as the quantization function ${\mathcal Q}(\cdot)$ are perfectly known by the BS. Then, the BS can  recover the channel matrix $\mathbf{H}$ from the coarsely quantized received signal $\mathbf{Y}$ by traditional channel estimation techniques, e.g., the linear minimum mean square error (LMMSE) estimation method \cite{LYZ}, the near maximum likelihood estimator \cite{nML}, and the approximate message passing approach \cite{MJH,XLYonebit}.  However, the traditional channel estimators need long pilot sequences and suffer from performance degradation due to the coarse nonlinear quantization error. 

\begin{figure}[!tbp]
\centering
\includegraphics[width=0.46\textwidth]{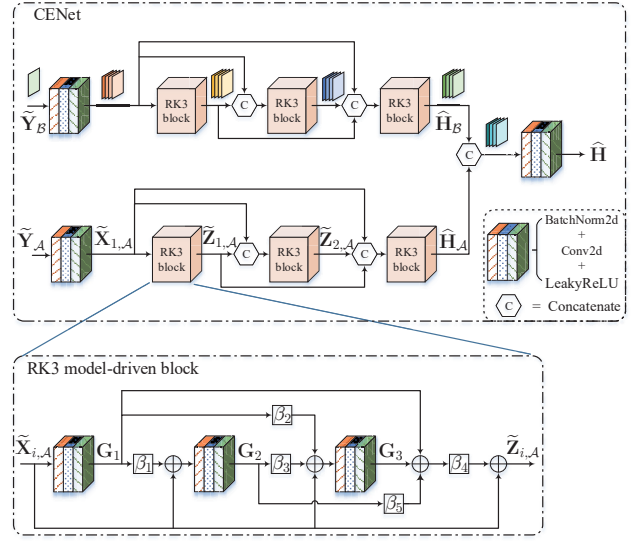}
\caption{The architecture of the proposed channel estimation network (CENet).}
\label{fig:CENet}
\end{figure}

On the other hand, DNNs are powerful tools to learning complex and latent relationship among data. Therefore, we resort to DL-based method to improve the channel estimation performance and reduce the pilot transmission overhead.
The architecture of the proposed channel estimation network (CENet) is depicted in Fig. \ref{fig:CENet}. The real and imaginary parts of the received signal $\mathbf{Y}$ are stacked as the input of CENet, which is denoted by $\widetilde{\mathbf{Y}}\in {\mathbb R}^{2M\times N_{p}}$:
\begin{equation}\label{eq:RxVectorReal}
\widetilde{\mathbf{Y}} = \begin{bmatrix} \Re\left(\mathbf{Y}\right)  \\ \Im\left(\mathbf{Y}\right) \end{bmatrix} \triangleq \mathcal{Q}\big( \widetilde{\mathbf{H}} \widetilde{\mathbf{P}} + \widetilde{\mathbf{W}} \big),
\end{equation}
where 
\begin{equation}
\widetilde{\mathbf{H}} = \begin{bmatrix} \Re\left(\mathbf{H}\right) &  -\Im\left(\mathbf{H}\right) \\ \Im\left(\mathbf{H}\right) & \Re\left(\mathbf{H}\right) \end{bmatrix},
\widetilde{\mathbf{P}} = \begin{bmatrix} \Re\left(\mathbf{P}\right)  \\ \Im\left(\mathbf{P}\right) \end{bmatrix},
\widetilde{\mathbf{W}} = \begin{bmatrix} \Re\left(\mathbf{W}\right)  \\ \Im\left(\mathbf{W}\right) \end{bmatrix}, \notag
\end{equation}
and ${\mathcal Q}(\cdot)$ reduces to real-valued element-wise quantization function. Next, we split $\widetilde{\mathbf{Y}}$ based on the quantization precision for the data processing stages, since it is more effective to feed data with different precision into different networks. We generate two antenna selection vectors $\mathbf{a}\in {\mathbb R}^{M}$ and $\mathbf{b}\in {\mathbb R}^{M}$ based on $\mathcal{A}$ and $\mathcal{B}$, respectively. The elements of $\mathbf{a}$ indexed by $\mathcal{A}$ are 1, and the remaining elements are 0. The same policy is applied for $\mathbf{b}$ and $\mathcal{B}$, which yields
\begin{equation}\label{eq:aandb}
[\mathbf{a}]_{m} = \left\{\begin{array}{ll}
1, & m \in \mathcal{A} \\
0, & m \notin \mathcal{A}
\end{array}\right. , \ 
[\mathbf{b}]_{m} = 1 - [\mathbf{a}]_{m} = \left\{\begin{array}{ll}
1, & m \in \mathcal{B} \\
0, & m \notin \mathcal{B}
\end{array}\right. .
\end{equation}
Then, by broadcasting $\mathbf{a}$ and $\mathbf{b}$ into matrices and element-wisely multiplied by the input of CENet $\widetilde{\mathbf{Y}}$, we obtain two matrices $\widetilde{\mathbf{Y}}_{\mathcal{A}}\in {\mathbb R}^{2M\times N_{p}}$ and $\widetilde{\mathbf{Y}}_{\mathcal{B}}\in {\mathbb R}^{2M\times N_{p}}$ respectively, which yields
\begin{equation}\label{eq:YaandYb}
\widetilde{\mathbf{Y}}_{\mathcal{A}} = \begin{bmatrix} \mathbf{a} &  \cdots & \mathbf{a} \\ \mathbf{a} &  \cdots & \mathbf{a} \end{bmatrix} \odot \widetilde{\mathbf{Y}}, \quad
\widetilde{\mathbf{Y}}_{\mathcal{B}} = \begin{bmatrix} \mathbf{b} &  \cdots & \mathbf{b} \\ \mathbf{b} &  \cdots & \mathbf{b} \end{bmatrix} \odot \widetilde{\mathbf{Y}}. 
\end{equation}
The matrices $\widetilde{\mathbf{Y}}_{\mathcal{A}}$ and $\widetilde{\mathbf{Y}}_{\mathcal{B}}$ contain full-resolution and one-bit quantized data respectively, and are fed into two parallel subnets of CENet. Since the two subnets share the same architecture, we use subscript-${\mathcal{A}}/{\mathcal{B}}$ to differentiate the parameters. 
In the following context, we will only focus on the parameters with subscript-${\mathcal{A}}$ for ease of notation, while  the similar discussion holds for  the parameters with subscript-${\mathcal{B}}$.

As illustrated in Fig. \ref{fig:CENet},  the first part of the subnet is a block to resize $\widetilde{\mathbf{Y}}_{\mathcal{A}}$ into the same shape of $\widetilde{\mathbf{H}}\in {\mathbb R}^{2M\times K}$, which involves batch normalization layer, 2D convolutional layer and LeakyReLU  activation layer. Note that the output size of this resizing block is $C_{1,\mathcal{A}}\times2M\times K$ where $C_{1,\mathcal{A}}$ is output features of the convolutional layers. Then, as depicted in Fig. \ref{fig:CENet}, the $C_{1,\mathcal{A}}\times2M\times K$ data is fed into densely connected blocks, which is inspired by DenseNet \cite{DenseNet}. Specifically, for each block, the outputs of all preceding blocks are concatenated along the first dimension and are used as inputs, while its outputs are passed to all subsequent blocks. The dense connection architecture has several advantages, e.g., requiring less computation to achieve high performance, and alleviating the vanishing-gradient problem especially for the one-bit quantized data. Nevertheless, the input size of the last block grows exponentially as the depth of the dense connection increases, which is computationally intensive and prevents us from using too many densely connected blocks, especially for massive MIMO scenario. The potential idea to address this issue is to use few densely connected blocks (3 blocks for our case) while design the inside architecture of each block, as illustrated in the following context.

Denote the input and output of the $i$th densely connected block as $\widetilde{\mathbf{X}}_{i,{\mathcal{A}}}$ and $\widetilde{\mathbf{Z}}_{i,{\mathcal{A}}}$, respectively. Due to the dense connection structure, we find that $\widetilde{\mathbf{X}}_{i,{\mathcal{A}}}$ and $\widetilde{\mathbf{Z}}_{i,{\mathcal{A}}}$ should have the same size as $C_{i,\mathcal{A}}\times2M\times K$. Therefore, the block can be interpreted as a decoding module that is widely used in computer vision tasks such as image super-resolution. Inspired by the super-resolution work in \cite{ODE}, we heuristically treat the decoding process as a dynamical model. Then, we adopt system of ordinary differential equations (ODE) to formulate the dynamical model, which yields
\begin{equation}\label{eq:ODE}
\left\{\begin{aligned}
 \frac{ \mathrm{d} \widetilde{\mathbf{X}}_{i}(t) } { \mathrm{d} t} &= f\left(\widetilde{\mathbf{X}}_{i}(t), t \right)  \\
 \widetilde{\mathbf{X}}_{i}(t_{0} ) &= \widetilde{\mathbf{X}}_{i,{\mathcal{A}}} \\
 \widetilde{\mathbf{X}}_{i}(t_{N}) &= \widetilde{\mathbf{Z}}_{i,{\mathcal{A}}}
\end{aligned}\right. ,
\end{equation}
where $f(\cdot,t):{\mathbb R}^{C_{i,\mathcal{A}}\times2M\times K} \rightarrow {\mathbb R}^{C_{i,\mathcal{A}}\times2M\times K}$ is the time dependence function, the initial value is $\widetilde{\mathbf{X}}_{i}(t_{0} )= \widetilde{\mathbf{X}}_{i,{\mathcal{A}}}$, and the final condition of the dynamical process is the decoded result as $\widetilde{\mathbf{X}}_{i}(t_{N}) = \widetilde{\mathbf{Z}}_{i,{\mathcal{A}}}$. 
% Note that there are total $C_{i,\mathcal{A}}\cdot2M\cdot K$ ODEs since it is a matrix form in \eqref{eq:ODE}. Hence the function $f(\cdot):{\mathbb R}^{C_{i,\mathcal{A}}\times2M\times K} \rightarrow {\mathbb R}^{1}$ has $C_{i,\mathcal{A}}\cdot2M\cdot K$ forms corresponded to the elements of $\widetilde{\mathbf{Z}}_{i}$, respectively, and we temporally abuse a single form $f(\cdot)$ without differentiation for ease of notation. 
To find the approximate solution of \eqref{eq:ODE}, we consider an iterative approach named as Runge-Kutta method to numerically solve ODEs with initial conditions. Specifically, the $3$rd order Runge-Kutta method for \eqref{eq:ODE} can be written as
% \begin{equation}\label{eq:3rdRunge-Kutta}
% \begin{gathered}
% \widetilde{\mathbf{X}}_{i}(t_{n+1}) = \widetilde{\mathbf{X}}_{i}(t_{n}) + \frac{1}{6}( \mathbf{S}_{1} + 4\mathbf{S}_{2} + \mathbf{S}_{3} )  \\
% \mathbf{S}_{1} = \mu f\left( \widetilde{\mathbf{X}}_{i}(t_{n}), t_{n} \right)  \\
% \mathbf{S}_{2} = \mu f\left( \widetilde{\mathbf{X}}_{i}(t_{n}) + \frac{1}{2} \mathbf{S}_{1}, t_{n} + \frac{1}{2} \mu \right)   \\
% \mathbf{S}_{3} = \mu f\left( \widetilde{\mathbf{X}}_{i}(t_{n}) - \mathbf{S}_{1} + 2 \mathbf{S}_{2}, t_{n} + \mu \right), 
% \end{gathered}
% \end{equation}
\begin{equation}\label{eq:3rdRunge-Kutta}
\begin{aligned}
&\widetilde{\mathbf{X}}_{i}(t_{n+1}) = \widetilde{\mathbf{X}}_{i}(t_{n}) + \frac{1}{6}( \mathbf{S}_{1} + 4\mathbf{S}_{2} + \mathbf{S}_{3} )  \\
&\mathbf{S}_{1} = \mu f\left( \widetilde{\mathbf{X}}_{i}(t_{n}), t_{n} \right)  \\
&\mathbf{S}_{2} = \mu f\left( \widetilde{\mathbf{X}}_{i}(t_{n}) + \frac{1}{2} \mathbf{S}_{1}, t_{n} + \frac{1}{2} \mu \right)   \\
&\mathbf{S}_{3} = \mu f\left( \widetilde{\mathbf{X}}_{i}(t_{n}) - \mathbf{S}_{1} + 2 \mathbf{S}_{2}, t_{n} + \mu \right), 
\end{aligned}
\end{equation}
where $\mu$ is the step size. 

Runge-Kutta method with higher order and smaller step size could provide better approximation for the solution of \eqref{eq:ODE}. Nevertheless, we cannot directly implement the Runge-Kutta method, since the function $f(\cdot)$ is unknown and needs to be designed as well. Then, for each block in the dense connection, neural networks are utilized to replace the unknown function $f(\cdot)$ and unfold the $3$rd order Runge-Kutta method into multi-layers structure, where a number of trainable parameters are introduced to replace the coefficients in \eqref{eq:3rdRunge-Kutta}. This technique is  called deep-unfolding or model-driven neural network and has also been used in \cite{JCESD,LJYDetect}. The proposed blocks are named as RK3 model-driven blocks, and the input-output relations of the $i$th densely connected RK3 model-driven block can be described as 
\begin{IEEEeqnarray}{l}
\widetilde{\mathbf{Z}}_{i,{\mathcal{A}}} = \widetilde{\mathbf{X}}_{i,{\mathcal{A}}} + \beta_{4} \left( \mathbf{G}_{1} + \beta_{5}\mathbf{G}_{2} + \mathbf{G}_{3} \right) \\
\mathbf{G}_{1} = \phi_{1}\left( \widetilde{\mathbf{X}}_{i,{\mathcal{A}}}, \boldsymbol{\Omega}_{1} \right) \\
\mathbf{G}_{2} = \phi_{2}\left( \widetilde{\mathbf{X}}_{i,{\mathcal{A}}} + \beta_{1} \mathbf{G}_{1}, \boldsymbol{\Omega}_{2} \right)  \\
\mathbf{G}_{3} = \phi_{3}\left( \widetilde{\mathbf{X}}_{i,{\mathcal{A}}} + \beta_{2} \mathbf{G}_{1} + \beta_{3} \mathbf{G}_{2}, \boldsymbol{\Omega}_{3} \right),
\end{IEEEeqnarray}
where $\beta_{1}$-$\beta_{5}$ are introduced trainable parameters, $\phi(\cdot, \boldsymbol{\Omega})$ is the input-output function of a neural network (involving batch normalization layer, 2D convolutional layer and LeakyReLU  activation layer). Moreover, $\boldsymbol{\Omega}$ denotes the weights and bias of the network. The structures of the RK3 model-driven blocks are illustrated in Fig. \ref{fig:CENet}.

By feeding the matrices $\widetilde{\mathbf{Y}}_{\mathcal{A}}$ and $\widetilde{\mathbf{Y}}_{\mathcal{B}}$ into two subnets consisted of densely connected RK3 model-driven blocks, we obtain two outputs denoted by $\widehat{\mathbf{H}}_{\mathcal{A}}\in{\mathbb R}^{C_{\mathcal{A}}\times2M\times K}$ and $\widehat{\mathbf{H}}_{\mathcal{B}}\in{\mathbb R}^{C_{\mathcal{B}}\times2M\times K}$, respectively. According to the features of dense connection, we know $C_{\mathcal{A}} = \sum_{i=1}^{3} C_{i,\mathcal{A}}$ and $C_{\mathcal{B}} = \sum_{i=1}^{3} C_{i,\mathcal{B}}$. Since $\widehat{\mathbf{H}}_{\mathcal{A}}$ and $\widehat{\mathbf{H}}_{\mathcal{B}}$ contain different information and have multiple features, we apply feature extraction  by concatenating $\widehat{\mathbf{H}}_{\mathcal{A}}$ and $\widehat{\mathbf{H}}_{\mathcal{B}}$ along the first dimension as the input of a fusion network. Ultimately, we obtain the estimated channel matrix $\widehat{\mathbf{H}}\in{\mathbb R}^{2M\times K}$ from the output of the fusion network. The structure of the fusion network and the whole architecture of CENet are depicted in Fig. \ref{fig:CENet}. 

The loss function of the proposed CENet is the normalized mean square error (NMSE) of the channel estimation, which is given by
\begin{equation}
\mathcal{L}_{\rm CENet} \left( \boldsymbol{\Theta}_{\rm CENet} \right) =  \frac{1}{VK} \sum_{v=1}^{V}\sum_{k=1}^{K} \frac{\left\|\tilde{\mathbf{h}}^{(v)}_{k}-\hat{\mathbf{h}}^{(v)}_{k}\right\|^{2}_{2}}{\left\|\tilde{\mathbf{h}}^{(v)}_{k}\right\|^{2}_{2}} ,
\end{equation}
where $\boldsymbol{\Theta}_{\rm CENet}$ contains all the trainable parameters of CENet, $V$ is the number of samples in each training batch, $\tilde{\mathbf{h}}_{k}$ and $\hat{\mathbf{h}}_{k}$ are the $k$th column vectors of $\widetilde{\mathbf{H}}$ and $\widehat{\mathbf{H}}$ respectively, and the superscript-($v$) is the index number in the training batch.

\section{Mixed-Precision ADCs Allocation} \label{sec:Mixed}
As mentioned in \eqref{eq:RxVector}, the function $\mathcal{Q}(\cdot)$ has multi-forms since mixed-resolution ADCs are deployed at the BS. Therefore, we next address the mixed-precision ADCs allocation problem: given the number of full-resolution and one-bit ADCs, how to maximize the channel estimation performance by properly assigning the mixed-precision ADCs to different antennas. Specifically, the vectors $\mathbf{a}$ and $\mathbf{b}$ should be optimized under the constraints $|\mathcal{A}|=M_{\mathcal{A}}$, $\mathcal{A} \cup \mathcal{B}=\{1,2, \ldots, M\}$, $\mathcal{A} \cap \mathcal{B}=\emptyset$ and \eqref{eq:aandb}. The joint mixed-precision ADCs allocation and channel estimation problem can be written as
\begin{equation}\label{eq:mixedADCopt}
\begin{array}{cl} 
{\min\limits_{\psi(\cdot), \mathcal{A}, \mathcal{B}}} & \mathbb{E}\left\{ \frac{1}{K}\sum_{k=1}^{K}\frac{\left\|\tilde{\mathbf{h}}_{k}-\hat{\mathbf{h}}_{k}\right\|^{2}_{2}}{\left\|\tilde{\mathbf{h}}_{k}\right\|^{2}_{2}} \right\} \\
{\text{ s.t.}} & [\hat{\mathbf{h}}_{1},\cdots,\hat{\mathbf{h}}_{K}] = \psi (\widetilde{\mathbf{P}}, \widetilde{\mathbf{Y}}_{\mathcal{A}}, \widetilde{\mathbf{Y}}_{\mathcal{B}}),\\
 & \eqref{eq:RxVectorReal}, \ \eqref{eq:aandb}\ {\rm and}\ \eqref{eq:YaandYb}, \\
 & |\mathcal{A}|=M_{\mathcal{A}},\ \mathcal{A} \cap \mathcal{B}=\emptyset, \\
 & \mathcal{A} \cup \mathcal{B}=\{1,2, \ldots, M\}, \\
\end{array}
\end{equation}
where $\psi(\cdot)$ is the channel estimator. A common way to solve \eqref{eq:mixedADCopt} is alternating optimization under the assumption that $\psi(\cdot)$ is linear, e.g., LMMSE estimator. However, the constraints in \eqref{eq:mixedADCopt} are non-convex, and the two subproblems  (one subproblem of channel estimator and the other of mixed-ADCs allocation) are still non-convex. Moreover, the subproblem of mixed-ADCs allocation is a combinatorial optimization problem, and it is hard to find the optimal solution of \eqref{eq:mixedADCopt} with traditional approaches.

Inspired by the idea of alternating optimization, we utilize two subnetworks to solve the two subproblems of  \eqref{eq:mixedADCopt}, and jointly optimize the two subnetworks. In particular, the CENet proposed in Section \ref{subsec:CENet} is adopted to find the nonlinear channel estimator $\psi(\cdot)$, and then a subnet named by selection network (SELNet) is designed to optimize the two  antenna selection vectors $\mathbf{a}$ and $\mathbf{b}$. The connections and cooperations between CENet and SELNet are depicted in Fig. \ref{fig:SELNetCENet}. Detailed design for SELNet is given in the following context.

\begin{figure}[!tbp]
\centering
\includegraphics[width=0.46\textwidth]{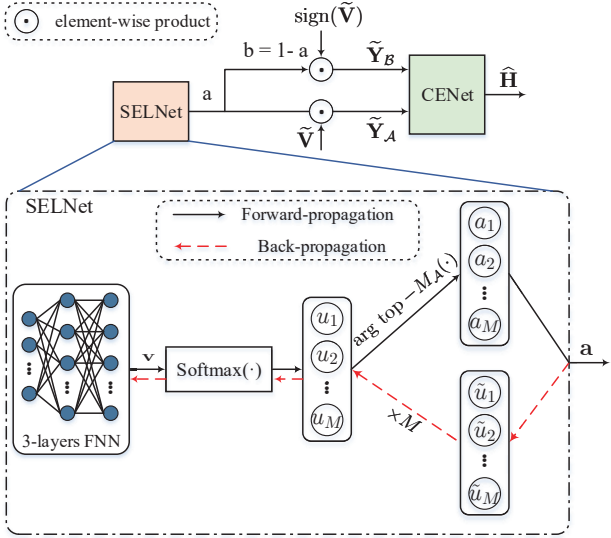}
\caption{The connections between CENet and SELNet, and the architecture of SELNet.}
\label{fig:SELNetCENet}
\end{figure}

From $|\mathcal{A}|=M_{\mathcal{A}}$, $\mathcal{A} \cup \mathcal{B}=\{1,2, \ldots, M\}$, $\mathcal{A} \cap \mathcal{B}=\emptyset$ and \eqref{eq:aandb}, we find that the vector $\mathbf{a}$ is $M_{\mathcal{A}}$-sparse and  $M_{\mathcal{A}}$-hot encoding. Therefore, the key of SELNet is to address the $M_{\mathcal{A}}$-hot encoding constraint and generate $\mathbf{a}$ with differentiable functions, which guarantees the successful back-propagation  during training stage. Specifically, SELNet first generates a probability distribution vector $\mathbf{u}\in\mathbb{R}^{M}$ with \emph{softmax} function and yields
\begin{equation}
[\mathbf{u}]_{m} = [\operatorname{softmax}(\mathbf{v})]_{m} = \frac{\exp\left([\mathbf{v}]_{m}\right)}{\sum_{i=1}^{M} \exp\left([\mathbf{v}]_{i}\right)},
\end{equation}
where the vector $\mathbf{v}\in\mathbb{R}^{M}$ is the output of the first block (A FNN composed of three layers) of SELNet as shown in Fig.~\ref{fig:SELNetCENet}. Note that each element of $\mathbf{u}$ lies in the interval $(0,1)$, and $[\mathbf{u}]_{m}$ can be interpreted as the probability of $[\mathbf{a}]_{m}$ being 1. Hence, a straightforward way to obtain the vector $\mathbf{a}$ is to find the $M_{\mathcal{A}}$ largest elements of $\mathbf{u}$ and set the corresponding elements in the vector $\mathbf{a}$ as 1, i.e.,
\begin{equation}\label{eq:topk}
\mathcal{A} = \operatorname{arg\ top}-M_{\mathcal{A}}(\mathbf{u}), \quad 
[\mathbf{a}]_{m} = \left\{\begin{array}{ll}
1, & m \in \mathcal{A} \\
0, & m \notin \mathcal{A}
\end{array}\right. ,
\end{equation}
where $\operatorname{arg\ top}-M_{\mathcal{A}}$ returns the indices of the $M_{\mathcal{A}}$ largest elements. Then, we can obtain the vector $\mathbf{b}$ by $\mathbf{b}=\mathbf{1} - \mathbf{a}$, as illustrated in Fig. \ref{fig:SELNetCENet}.

With $\operatorname{arg\ top}-M_{\mathcal{A}}$ operator, SELNet now generates the vectors $\mathbf{a}$ and $\mathbf{b}$. However, the operator in \eqref{eq:topk} is non-differentiable, which  stymies the back-propagation algorithm when training the networks. As a result, the neural layers before this non-differentiable operator cannot be trained by the back-propagation algorithm. To tackle this issue, we use \eqref{eq:topk} for forward-propagation while replace \eqref{eq:topk} with a differentiable operator during the back-propagation. Such replacement becomes exact when some constraints  are  satisfied. Specifically, during back-propagation, we replace the vector $\mathbf{a}$ with $\tilde{\mathbf{u}}$  defined as
\begin{equation}\label{eq:apprxtopk}
[\tilde{\mathbf{u}}]_{m} = M_{\mathcal{A}} [\mathbf{u}]_{m} = M_{\mathcal{A}} \frac{\exp\left([\mathbf{v}]_{m}\right)}{\sum_{i=1}^{M} \exp\left([\mathbf{v}]_{i}\right)},
\end{equation}
where $[\tilde{\mathbf{u}}]_{m}$ lies in the interval $(0,M_{\mathcal{A}})$, and $\sum_{i=1}^{M}[\tilde{\mathbf{u}}]_{i} = M_{\mathcal{A}}$. The forward, back-propagation and the structures of SELNet are illustrated in Fig. \ref{fig:SELNetCENet}. It is worth pointing out that the operator in \eqref{eq:apprxtopk} is differentiable. Moreover, $\tilde{\mathbf{u}}$  equals to the vector $\mathbf{a}$ with some constraints as stated in the following Theorem:
\begin{theorem}\label{theorem:K-hot}
If a vector $\mathbf{x}=\{[x_{1}, x_{2}, \cdots, x_{M}]|\forall m, x_{m}\in\mathbb{R}, x_{m} \geq 0 \}$ satisfies the following constraints:
\begin{IEEEeqnarray}{c}
\| \mathbf{x} \|^{r}_{r} = x^{r}_{1} + x^{r}_{2} + \cdots + x^{r}_{M} = K, \\
\| \mathbf{x} \|^{p}_{p} = x^{p}_{1} + x^{p}_{2} + \cdots + x^{p}_{M} = K, \\
\| \mathbf{x} \|^{q}_{q} = x^{q}_{1} + x^{q}_{2} + \cdots + x^{q}_{M} = K, \\
K \leq M, \quad 0<r<p<q<\infty, 
\end{IEEEeqnarray}
then the vector $\mathbf{x}$ is $K$-hot encoding vector, i.e., $K$  elements of $\mathbf{x}$ are `1' and $(M-K)$ elements of $\mathbf{x}$ are `0'.
\end{theorem}
\begin{proof}
% The proof can be found in Appendix \ref{app:1}. 
Without loss of generality, we assume the elements of the vector $\mathbf{x}$ are sorted in descending order as $\mathbf{x}=[x_{1},x_{2},\cdots,x_{L},0,\cdots,0]$, where $x_{1}>x_{2}>\cdots,x_{L}>0$ and $L$ denotes the number of non-zero elements in $\mathbf{x}$. Then, the constraints of the vector $\mathbf{x}$  can be rewritten as 
\begin{IEEEeqnarray}{c}
\| \mathbf{x} \|^{r}_{r} = x^{r}_{1} + x^{r}_{2} + \cdots + x^{r}_{L} = K, \label{eq:constrtLr}\\
\| \mathbf{x} \|^{p}_{p} = x^{p}_{1} + x^{p}_{2} + \cdots + x^{p}_{L} = K, \label{eq:constrtLp}\\
\| \mathbf{x} \|^{q}_{q} = x^{q}_{1} + x^{q}_{2} + \cdots + x^{q}_{L} = K, \label{eq:constrtLq}\\
K \leq M, \quad 0<r<p<q<\infty.
\end{IEEEeqnarray}
Additionally, we set $c= \frac{q-r}{q-p} $ and $d=\frac{q-r}{p-r}$, which yields $c,d\in(1,\infty)$ and 
\begin{equation}\label{eq:designedfactor}
\frac{1}{c} + \frac{1}{d} = 1, \quad \frac{r}{c} + \frac{q}{d} = p.
\end{equation}

Now we can apply H\"{o}lder's inequality for sums, which states that
\begin{equation}\label{eq:holderineq}
\sum\limits_{i=1}^{L}|y_{i}z_{i}| \leq \left( \sum\limits_{i=1}^{L}|y_{i}|^{c} \right)^{\frac{1}{c}} \left( \sum\limits_{i=1}^{L}|z_{i}|^{d} \right)^{\frac{1}{d}},
\end{equation}
where $c,d\in(1,\infty)$ and $\frac{1}{c} + \frac{1}{d} = 1$.  The equality in \eqref{eq:holderineq} holds if and only if the vectors $[y_{1},\cdots,y_{L}]^{T}$ and $[z_{1},\cdots,z_{L}]^{T}$ are linearly dependent.

Hence, we choose $y_{i}=x^{r/c}_{i}$ and $z_{i}=x^{q/d}_{i}$, and substitute $y_{i}$ and $z_{i}$ into H\"{o}lder's inequality \eqref{eq:holderineq}. Then, we obtain
\begin{equation}\label{eq:holderx}
\sum\limits_{i=1}^{L} x^{\frac{r}{c}}_{i} x^{\frac{q}{d}}_{i} \leq \left( \sum\limits_{i=1}^{L}(x^{\frac{r}{c}}_{i})^{c} \right)^{\frac{1}{c}} \left( \sum\limits_{i=1}^{L}(x^{\frac{q}{d}}_{i})^{d} \right)^{\frac{1}{d}}.
\end{equation}
Substituting \eqref{eq:designedfactor} into \eqref{eq:holderx} and after some algebra, we obtain
\begin{equation}\label{eq:ineqtoeq}
\sum\limits_{i=1}^{L} x^{p}_{i}  \leq \left( \sum\limits_{i=1}^{L}x^{r}_{i} \right)^{\frac{1}{c}} \left( \sum\limits_{i=1}^{L}x^{q}_{i} \right)^{\frac{1}{d}}.
\end{equation}
Note that the inequality in \eqref{eq:ineqtoeq} becomes an equality with the constraints \eqref{eq:constrtLr}-\eqref{eq:constrtLq}. Therefore, the two vectors $[x^{r/c}_{1},\cdots,x^{r/c}_{L}]^{T}$ and $[x^{q/d}_{1},\cdots,x^{q/d}_{L}]^{T}$ must be linearly dependent. With the linear dependence condition as well as the constraints \eqref{eq:constrtLr}-\eqref{eq:constrtLq}, and after some algebra,  we can easily obtain 
\begin{IEEEeqnarray}{l}
L = K, \\
x_{1} = x_{2} = \cdots = x_{L} = 1,
\end{IEEEeqnarray}
which indicates that the vector $\mathbf{x}$ is $K$-hot encoding. 
This completes the proof.
\end{proof}

Note that the vector $\tilde{\mathbf{u}}$ in \eqref{eq:apprxtopk} satisfies $\ell_{1}$-norm constraint as $\| \tilde{\mathbf{u}} \|_{1} = M_{\mathcal{A}} $. According  to Theorem \ref{theorem:K-hot}, another two $\ell_{p}$-norm constraints are needed to ensure $\tilde{\mathbf{u}}$ being equal to the vector $\mathbf{a}$. Hence, for ease of computation and notation, we utilize $\ell_{2}$-norm and $\ell_{3}$-norm  with $\| \tilde{\mathbf{u}} \|^{2}_{2} = M_{\mathcal{A}} $ and $\| \tilde{\mathbf{u}} \|^{3}_{3} = M_{\mathcal{A}} $ as the two constraints,  which has also been adopted in \cite{ASNLinBo}. However, it is difficult to simultaneously  embed both $\ell_{2}$-norm and $\ell_{3}$-norm constraints for $\tilde{\mathbf{u}}$ in a neural network. Alternatively, we can embed the constraints into the loss function of the network so that the constraints are gradually satisfied during training process. Specifically, the loss function of SELNet is given by 
\begin{equation}\label{eq:SELNetLoss}
\mathcal{L}_{\rm SELNet}\left( \boldsymbol{\Theta}_{\rm SELNet} \right) = \gamma_{1}( \| \tilde{\mathbf{u}} \|^{2}_{2} - M_{\mathcal{A}} )^{2} + \gamma_{2}( \| \tilde{\mathbf{u}} \|^{3}_{3} - M_{\mathcal{A}} )^{2},
\end{equation} 
where $\boldsymbol{\Theta}_{\rm SELNet}$ contains all the trainable parameters of SELNet, and the hyperparameters $\gamma_{1}$ and $\gamma_{2}$ are to balance the two terms.  During training stage,  both $\| \tilde{\mathbf{u}} \|^{2}_{2} - M_{\mathcal{A}}$ and $\| \tilde{\mathbf{u}} \|^{3}_{3} - M_{\mathcal{A}}$ gradually reduce to zero when the loss $\mathcal{L}_{\rm SELNet}$ keeps decreasing. Therefore, the approximation between $\tilde{\mathbf{u}}$  and the vector $\mathbf{a}$ progressively becomes exact.

The forward and back-propagation of SELNet, the connections, and the cooperations between CENet and SELNet are summarized in Fig. \ref{fig:SELNetCENet}. The overall loss function is 
\begin{equation}\label{eq:MixedLoss}
\mathcal{L}_{\rm CENet} \left( \boldsymbol{\Theta}_{\rm CENet} \right) + \gamma_{3} \mathcal{L}_{\rm SELNet}\left( \boldsymbol{\Theta}_{\rm SELNet} \right),
\end{equation} 
where $\gamma_{3}$ is a hyperparameter. By jointly training CENet and SELNet with the loss \eqref{eq:MixedLoss}, both the nonlinear channel estimator $\psi(\cdot)$ and the selection vector will be obtained.
 
\section{Joint Pilot Design and Channel Estimation} \label{sec:Pilot}
Note that  both  Section \ref{sec:UL CE} and \ref{sec:Mixed} address the channel estimation problem at the receiver side, while the pilot design at the transmitter side has not been considered yet. By learning the correlations and statistics of the wireless channel, we can optimize the pilot signal to reduce the pilot overhead and improve the channel estimation accuracy significantly \cite{AhmedPilot}. 

The joint pilot design and channel estimation with mixed-ADCs allocation problem can be formulated as
\begin{equation}\label{eq:jointPilotCE}
\begin{array}{cl} 
{\min\limits_{ \psi(\cdot),  \widetilde{\mathbf{P}}, \mathcal{A}, \mathcal{B}}} & \mathbb{E}\left\{ \frac{1}{K}\sum_{k=1}^{K}\frac{\left\|\tilde{\mathbf{h}}_{k}-\hat{\mathbf{h}}_{k}\right\|^{2}_{2}}{\left\|\tilde{\mathbf{h}}_{k}\right\|^{2}_{2}} \right\} \\
{\text{ s.t.}} & [\hat{\mathbf{h}}_{1},\cdots,\hat{\mathbf{h}}_{K}] = \psi (\widetilde{\mathbf{P}}, \widetilde{\mathbf{Y}}_{\mathcal{A}}, \widetilde{\mathbf{Y}}_{\mathcal{B}}),\\
 & \eqref{eq:RxVectorReal}, \ \eqref{eq:aandb}\ {\rm and}\ \eqref{eq:YaandYb}, \\
 & |\mathcal{A}|=M_{\mathcal{A}},\ \mathcal{A} \cap \mathcal{B}=\emptyset, \\
 & \mathcal{A} \cup \mathcal{B}=\{1,2, \ldots, M\}, \\
 & \operatorname{tr}(\widetilde{\mathbf{P}}\widetilde{\mathbf{P}}^{T}) = N_{p}\rho, 
\end{array}
\end{equation}
where $\rho$ denotes the average transmission power,  and $\psi(\cdot)$ is the channel estimator.
Similar to the optimization problem in  \eqref{eq:mixedADCopt}, it is hard to solve \eqref{eq:jointPilotCE} by traditional signal processing methods due to the non-convexity and the combinatorial optimization problem. Furthermore, it is worth noting that a new solution of \eqref{eq:jointPilotCE} needs to be obtained whenever the channel $\widetilde{\mathbf{H}}$ changes, which hampers the effectiveness of the traditional solver. To tackle this issue, we utilize DL-based method to learn the features  over a channel dataset and obtain the corresponding jointly optimized pilot and channel estimator. Specifically, we devise a subnet named by pilot design network (PDNet) to cooperate with the aforementioned CENet and SELNet, and then develop a autoencoder-inspired end-to-end architecture to jointly optimize the three subnets, as depicted in Fig. \ref{fig:JointNet}

\begin{figure}[!tbp]
\centering
\includegraphics[width=0.46\textwidth]{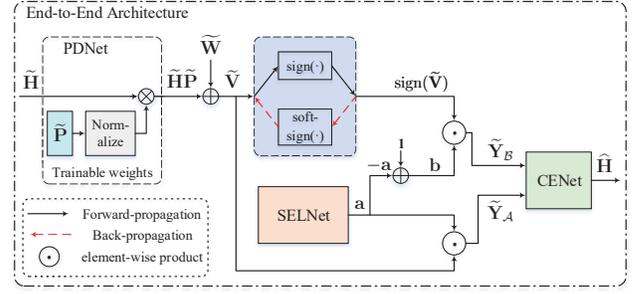}
\caption{The autoencoder-inspired end-to-end architecture for joint optimization.}
\label{fig:JointNet}
\end{figure}

From \eqref{eq:RxVectorReal}, we propose a neural network called PDNet to realize the the matrix multiplication between $\widetilde{\mathbf{H}}$ and $\widetilde{\mathbf{P}}$.
In particular, if the input of the network is $\widetilde{\mathbf{H}}$, and the trainable weights of the network are $\widetilde{\mathbf{P}}$, and then  the network outputs $\widetilde{\mathbf{H}}\widetilde{\mathbf{P}}$. Furthermore, to embed the transmission power constraint $\operatorname{tr}(\widetilde{\mathbf{P}}\widetilde{\mathbf{P}}^{T}) = N_{p}\rho$ into the network, a normalization layer is applied to $\widetilde{\mathbf{P}}$ prior to the matrix multiplication operation, which is given by  
\begin{equation}\label{eq:PowerNormalized}
\sqrt{N_{p}\rho}\frac{\widetilde{\mathbf{P}}}{\| \widetilde{\mathbf{P}} \|_{F}}.
\end{equation}
The structure of PDNet is illustrated in Fig. \ref{fig:JointNet}. Note that the trainable parameters of PDNet are optimized and obtained by training the networks over dataset, and namely the optimized pilot matrix $\widetilde{\mathbf{P}}$ is directly obtained from the thoroughly trained weights of PDNet. Therefore, the key of the joint pilot design and channel estimation with mixed-ADCs allocation problem is  to  integrate PDNet, CENet and SELNet into a whole architecture for the sake of joint optimization.

\subsection{Autoencoder Inspired End-to-End Architecture}\label{sec:end-to-end}
Note that the pilot signal is transmitted by the users, while the channel  will have to be estimated  at the BS. This indicates that the joint optimization of PDNet, CENet and SELNet should be implemented in an end-to-end manner, where the transmitter, the receiver, the process of transmission and the quantization should be involved and emulated. The autoencoder \cite{DLbook} is an end-to-end optimized neural network where the input is copied to the output, and some hidden layers called encoder and decoder are used to learn the latent representations. Hence, we will develop an end-to-end network structure inspired by the  architecture of autoencoder.  The analogy  is given as follows:
\begin{enumerate*} \item the channel matrix is the input; \item the proposed PDNet denotes the encoder to learn the latent features of the channel matrix; \item a noise layer and a quantization layer are appended to mimic the noise and quantizer of the receiver; \item the proposed SELNet and CENet are  regarded collectively as the decoder to process the quantized received signal; \item the output is the estimated channel matrix.  \end{enumerate*}
The structure of the end-to-end network  is depicted in Fig. \ref{fig:JointNet}. 

We can then obtain the jointly optimized pilot, channel estimator and selection vectors by training the end-to-end network with the following aggregate loss function 
\begin{multline}\label{eq:overallloss}
\mathcal{L} (\boldsymbol{\Theta}_{\rm PDNet},\boldsymbol{\Theta}_{\rm CENet},\boldsymbol{\Theta}_{\rm SELNet}) = \mathcal{L}_{\rm CENet} \left( \boldsymbol{\Theta}_{\rm CENet} \right) \\ + \gamma_{3} \mathcal{L}_{\rm SELNet}\left( \boldsymbol{\Theta}_{\rm SELNet} \right),
\end{multline}
where $\boldsymbol{\Theta}_{\rm PDNet}$ denotes the trainable weights of PDNet. However, the quantization layer hinders the back-propagation of the end-to-end network, because the quantization function $\operatorname{sign}(\cdot)$ is not differentiable at the origin and has zero derivative everywhere else. Consequently, the proposed PDNet that is prior to the quantization layer cannot be trained by the back-propagation algorithm.

To address this issue, we can use a differentiable function to approximate the $\operatorname{sign}(\cdot)$ function. In \cite{YW1b}, the \emph{sigmoid} function with adjustable slope has been utilized to replace $\operatorname{sign}(\cdot)$, where the steepness of the sigmoid function is slowly increased during training progress. Specifically, the sigmoid function with adjustable slope is denoted by 
\begin{equation}\label{eq:adjustableSigmoid}
\operatorname{softsign}(x) = 2\operatorname{sigmoid}(\kappa x) - 1 = \frac{2}{ 1 + \exp(-\kappa x) } - 1,
\end{equation}
where $\kappa$ is the factor to adjust the steepness.  The  input-output relationships of the sigmoid function with different value of $\kappa$ are illustrated in Fig. \ref{fig:AdjustableSigmoid}, and it can be found that increasing  $\kappa$ will make the sigmoid function steeper. Here we propose two choices to utilize the adjustable sigmoid function: 
\begin{enumerate*}
\item replacing $\operatorname{sign}(\cdot)$ with the adjustable sigmoid for both forward and back-propagation; \item using $\operatorname{sign}(\cdot)$ for forward-propagation and the adjustable sigmoid for back-propagation.
\end{enumerate*}
Now the non-differentiation problem is addressed by using the adjustable sigmoid function. Next, we focus on the training and implementation strategies for the whole end-to-end network.
\begin{figure}[!tbp]
\centering
\includegraphics[width=0.38\textwidth]{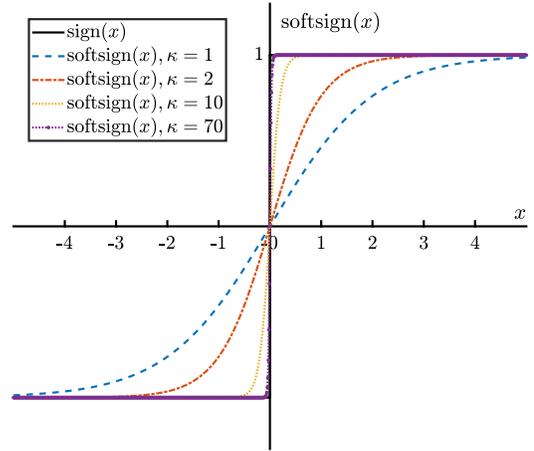}
\caption{The The  input-output relationships of softsign function with different values of $\kappa$.}
\label{fig:AdjustableSigmoid}
\end{figure}

\subsection{Off-line Training and On-line Implementation}

\begin{algorithm}[!tb]
\caption{Training and deployment of the end-to-end network}
\label{alg:Training}
\begin{algorithmic}[1]
% \REQUIRE  $\widetilde{\mathbf{H}}$
\STATE {\bf initialize:} Training dataset $\mathcal{D}$ of $\widetilde{\mathbf{H}}$, the number of iterations $N_{\rm iter}$, hyperparameters $C_{1,\mathcal{{A}}}$-$C_{3,\mathcal{{A}}}$, $C_{1,\mathcal{{B}}}$-$C_{3,\mathcal{{B}}}$ and $\kappa$,  and trainable parameters $\boldsymbol{\Theta}_{\rm PDNet}$, $\boldsymbol{\Theta}_{\rm CENet}$, $\boldsymbol{\Theta}_{\rm SELNet}$.
\STATE {\bf On-line Training:}
\FOR{$i=1:N_{\rm iter}$}
\item[] \COMMENT{forward-propagation}
\STATE Draw a random subset of $\mathcal{D}$ to generate a mini-batch 
\STATE Feed $\widetilde{\mathbf{H}}$ into PDNet to obtain: $\widetilde{\mathbf{H}}\widetilde{\mathbf{P}}$
\STATE The noise layer outputs: $\widetilde{\mathbf{Z}} = \widetilde{\mathbf{H}}\widetilde{\mathbf{P}} + \widetilde{\mathbf{W}}$
\STATE The quantization layer outputs: $\operatorname{sign}(\widetilde{\mathbf{Z}})$
\STATE SELNet generates: $\mathbf{a}$, $\mathbf{b}=\mathbf{1} - \mathbf{a}$ and $\tilde{\mathbf{u}}$
\STATE Compute the quantized signals: $\widetilde{\mathbf{Y}}_{\mathcal{A}} = \begin{bmatrix} \mathbf{a} &  \cdots & \mathbf{a} \\ \mathbf{a} &  \cdots & \mathbf{a} \end{bmatrix} \odot \widetilde{\mathbf{Z}}, \ \widetilde{\mathbf{Y}}_{\mathcal{B}} = \begin{bmatrix} \mathbf{b} &  \cdots & \mathbf{b} \\ \mathbf{b} &  \cdots & \mathbf{b} \end{bmatrix} \odot \operatorname{sign}(\widetilde{\mathbf{Z}}) $
\STATE Input $\widetilde{\mathbf{Y}}_{\mathcal{A}}$ and $\widetilde{\mathbf{Y}}_{\mathcal{B}}$ into CENet and estimate:  $\widehat{\mathbf{H}}$
\STATE Compute the loss: $\mathcal{L} (\boldsymbol{\Theta}_{\rm PDNet},\boldsymbol{\Theta}_{\rm CENet},\boldsymbol{\Theta}_{\rm SELNet})$
\item[] \COMMENT{back-propagation}
\STATE Replace $\operatorname{sign}(\cdot)$ with \eqref{eq:adjustableSigmoid} for the quantization layer
\STATE Replace the vector $\mathbf{a}$ with \eqref{eq:apprxtopk} for SELNet
\STATE Apply back-propagation algorithm and update $\boldsymbol{\Theta}_{\rm PDNet}$, $\boldsymbol{\Theta}_{\rm CENet}$ and $\boldsymbol{\Theta}_{\rm SELNet}$ by Adam optimizer
\ENDFOR
\STATE Acquire the optimized pilot from the weights of PDNet 
\STATE Obtain the optimized vectors $\mathbf{a}$ and $\mathbf{b}$ from SELNet 
\STATE Acquire the optimized CENet
\STATE {\bf Off-line Deployment:}
\STATE Assign the pilot to the users for the uplink transmission
\STATE Allocate the mixed-ADCs at the BS based on $\mathbf{a}$ and $\mathbf{b}$
\STATE Deploy CENet at the BS to estimate $\widehat{\mathbf{H}}$ from $\widetilde{\mathbf{Y}}_{\mathcal{A}}$ and $\widetilde{\mathbf{Y}}_{\mathcal{B}}$
\ENSURE  $\widehat{\mathbf{H}}$
\end{algorithmic} 
\end{algorithm}
The whole end-to-end network architecture for joint pilot design and channel estimation with mixed-ADCs allocation is demonstrated in Fig. \ref{fig:JointNet}. The training of the whole end-to-end network can be performed in an off-line manner (e.g.,  cloud computing and remote server) with sufficient dataset of the channel matrix. When the whole network is trained and  optimized, the network is then properly split and assigned to the transmitters and receivers, respectively. The details are given as follows:

During off-line training stage, the data-flow diagram for the forward-propagation is illustrated in Fig. \ref{fig:JointNet}. Notice that in the back-propagation, the vector $\mathbf{a}$ is replaced by $\widetilde{\mathbf{u}}$ for the proposed SELNet, and $\operatorname{sign}(\cdot)$ is replaced by the adjustable sigmoid for the quantization layer.
Moreover, the adaptive moment estimation (Adam) optimizer \cite{adam} is utilized in the back-propagation to update the trainable parameters $\boldsymbol{\Theta}_{\rm PDNet}$, $\boldsymbol{\Theta}_{\rm CENet}$ and $\boldsymbol{\Theta}_{\rm SELNet}$ by minimizing the loss functions in \eqref{eq:overallloss}. 

After the whole end-to-end network is trained, it can be implemented in practical systems. Specifically, we obtain the optimized pilot directly from the weights of PDNet and assign the pilot to the users. Then, the BS allocate the mixed-ADCs according to the selection vectors $\mathbf{a}$ and $\mathbf{b}$ that are optimized by SELNet. Furthermore, the optimized CENet is deployed at the BS to estimate the channel matrix from the quantized received signal. Additionally, the network will be retrained and re-implemented until the statistics of the wireless channel has  changed significantly.

The detailed steps of training and implementation are summarized in Algorithm \ref{alg:Training}.

\section{Numerical Results} \label{sec:sim}
In this section, we present simulation results to show the effectiveness of the proposed joint pilot design and channel estimation method, and make comparisons with the state-of-the-art algorithms. 
\subsection{Dataset generation}
We consider that the BS has a ULA with $M=64$ antennas serving $K=8$ users, where each user has $L_{k} = 3$  channel paths.  
The initial values of all the $24$ DOAs are randomly generated within the interval $[-80^{\circ},80^{\circ}]$. The initial values of the channel complex gains are randomly generated with independent and identically distributed as $\mathcal{CN}(0,1)$. 
To generate the dataset, small perturbations are added to the initial values of the DOAs and channel gains. Specifically, perturbations randomly generated within the interval $[-4^{\circ},4^{\circ}]$ are added to the initial DOAs, and perturbations randomly generated with $\mathcal{CN}(0,0.04)$ are added to the initial channel gains. The generated datasets of $\widetilde{\mathbf{H}}$ are randomly divided into training and testing sets with $100,000$ and $5,000$ samples, respectively.
\subsection{Network Hyperparameters Configuration}
Adam optimizer is  adopt for the end-to-end network training with batch size $100$. The initial learning rate is set as $2 \times 10^{-3}$ and exponential decay every $20$ epochs. The decay factor and the maximum number of epochs are set as $0.7$ and $200$, respectively. 

The output features of the densely connected RK3 model-driven blocks in CENet are denoted by $C_{1,\mathcal{{A}}}$-$C_{3,\mathcal{{A}}}$ and $C_{1,\mathcal{{B}}}$-$C_{3,\mathcal{{B}}}$. The output features are set as $C_{1,\mathcal{{A}}} = 60$, $C_{2,\mathcal{{A}}} = 120$, $C_{3,\mathcal{{A}}} = 240$, $C_{1,\mathcal{{B}}} = 20$, $C_{2,\mathcal{{B}}} = 40$ and $C_{3,\mathcal{{B}}} = 80$. 

The layer size of the 3-layers FNN in SELNet is set as $16$, $32$ and $64$, respectively, and  LeakyReLU activation function is used in the the 3-layers FNN.
The hyperparameters of the loss function of SELNet in \eqref{eq:SELNetLoss} are set as $\gamma_{1}=\gamma_{2}=1$. The penalty factor $\gamma_{3}$ in \eqref{eq:MixedLoss} is initially set as $\gamma_{3}=0.01$ and $0.02$ increment every epoch with maximum value $0.5$. 

The pilot length $N_{p}$ is a hyperparameter of the trainable weights $\widetilde{\mathbf{P}}\in {\mathbb C}^{2K\times N_{p}}$ in PDNet, and we will investigate the performances of the proposed method under different values of $N_{p}$. The steepness factor $\kappa$ in  \eqref{eq:adjustableSigmoid} is set as $\kappa = 70$ for better approximation of $\operatorname{sign}(\cdot)$. 

\subsection{Evaluation Metrics}
The performance metrics of channel estimation are defined as the normalized mean square error (NMSE)
\begin{IEEEeqnarray}{l}
\mathrm{NMSE} \triangleq \mathbb{E}\left\{ \frac{1}{K} \sum_{k=1}^{K} \frac{\left\|\tilde{\mathbf{h}}_{k}-\hat{\mathbf{h}}_{k}\right\|^{2}_{2}}{\left\|\tilde{\mathbf{h}}_{k}\right\|^{2}_{2}} \right\}, \label{eq:NMSEdef}
\end{IEEEeqnarray}
where $\tilde{\mathbf{h}}_{k}$ and $\hat{\mathbf{h}}_{k}$ are the $k$th column vectors of $\widetilde{\mathbf{H}}$ and $\widehat{\mathbf{H}}$ respectively. 
The channel is normalized, and SNR is defined as 
\begin{equation}
\mathrm{SNR} \triangleq 10\log_{10}\left( \frac{\| \widetilde{\mathbf{P}} \|^{2}_{F}}{N_p \sigma^2} \right)\quad \mathrm{dB},
\end{equation}
where the noise variance $\widetilde{\mathbf{W}}$ is set as $\sigma^2 = 1$. According to \eqref{eq:PowerNormalized}, SNR equals to the ratio of the average transmission power $\rho$ to $\sigma^2$ ($\mathrm{SNR} = 10\log_{10}\left( \frac{\rho}{\sigma^2} \right)$) when the channel is normalized.

The performances of the following channel estimation algorithms will be taken as comparisons:
\begin{itemize}
\item The compressed sensing based gridless generalized approximate message passing (GL-GAMP) algorithm  that exhibits superior performance over the state-of-the-art channel estimation algorithms \cite{XLYonebit}. The pilot sequence is chosen as a length-$N_{p}$ Zadoff-Chu (ZC) sequence, and each row of  the pilot matrix $\mathbf{P}$ is a circularly shifted version of the ZC sequence, which assures the orthogonality of the pilot sequences for different users. Moreover, all ADCs are one-bit precision.
\item "CENet, 1-bit": the pilot sequence is ZC sequence, and all ADCs are one-bit precision. We only use CENet at the BS to estimate the channel matrix.
\item "PDNet+CENet, 1-bit": all ADCs are one-bit precision. We use PDNet and CENet with the end-to-end architecture for joint pilot design and channel estimation.
\item "PDNet+CENet, fixed $M_{\mathcal{A}}=\#$": there are $M_{\mathcal{A}}$ full-resolution ADCs. We omit SELNet and use a fixed allocation strategy for the mixed-precision ADCs, i.e., the set $\mathcal{A}$ is fixed and unoptimized.
\item We replace the densely connected RK3 model-driven blocks in CENet with regular CNNs, which is a modification of the algorithm in \cite{DNNCEPD}. We add the term "CNN" to denote this case.
\end{itemize}
The abbreviations and configurations of the comparison algorithms are summarized in Table \ref{table:benchmarks}.
\begin{table*}[!tbh]
\centering
\caption{Abbreviations and configurations of the comparison algorithms}
\label{table:benchmarks}
\begin{tabular}{cccccc}
\toprule
% abbreviations & pilot design & one-bit ADCs & full-resolution ADCs & ADCs allocation & channel estimation \\
algorithm & pilot & one-bit & full-resolution & ADCs & channel \\
abbreviation & design & ADCs & ADCs & allocation & estimation \\
\midrule
GAMP & no, ZC  & $M$  & $0$ & - & GAMP\cite{XLYonebit} \\
\cmidrule{1-6}
GL-GAMP & no, ZC  & $M$  & $0$ & - & GL-GAMP\cite{XLYonebit} \\
\cmidrule{1-6}
CENet, 1-bit & no, ZC  & $M$  & $0$ & - & CENet \\
\cmidrule{1-6}
PDNet+CENet, 1-bit & yes, PDNet  & $M$  & $0$ & - & CENet \\
\cmidrule{1-6}
PDNet+CENet, fixed $M_{\mathcal{A}}=\#$ & yes, PDNet  & $M-\#$  & $\#$ & fixed & CENet \\
\cmidrule{1-6}
CNN & -  & - & - & - & regular CNNs \\
\cmidrule{1-6}
proposed, $M_{\mathcal{A}}=\#$ & yes, PDNet  & $M-\#$ & $\#$ & SELNet & CENet \\
\bottomrule
\end{tabular}
\end{table*}
\subsection{Performance Analysis}
In Fig.~\ref{fig:NMSE16}, the NMSE of  channel estimate is plotted over SNR, where $N_{p} = 64$. 
Note that GAMP only considers unstructured channel model without DOA estimation, while GL-GAMP \cite{XLYonebit} estimates DOAs as well as channel gains and then reconstruct the angular domain channel. Hence, GAMP has poor performance than GL-GAMP.  
Compared with GL-GAMP, "CENet, 1-bit" algorithm has the same configurations about pilot and quantization resolution but utilize a DL-based channel estimation method. We observe that "CENet, 1-bit" outperforms GL-GAMP with more than $3$ dB performance gap in the high SNR regime, which shows the superiority of the proposed CENet over the traditional channel estimator with one-bit quantization.
Moreover, it can be seen that there is considerable performance gap between "CENet, 1-bit" and "PDNet+CENet, 1-bit", which shows the effectiveness of the proposed pilot design approach. 
This indicates that the proposed PDNet can optimize the pilot to capture the intrinsic characteristics of the sparse mmWave channels,  and then significantly enhance the channel estimation performance by collaborated with CENet.
Furthermore,  the results also show that "PDNet+CENet, fixed $M_{\mathcal{A}}=16$" has notable performance improvement over "PDNet+CENet, 1-bit", since $16$ full-resolution ADCs are utilized.
Hence, one effective way to increase the estimation accuracy is to increase the number of high-resolution ADCs with mixed-precision architecture.
Here, the $16$ full-resolution ADCs are equispaced within the $M=64$ antennas for the "PDNet+CENet, fixed $M_{\mathcal{A}}=16$" case, whereas the equispaced deployment is not necessarily optimal. 
As a comparison, we illustrate the performance of the proposed  end-to-end architecture with SELNet that will optimize the selection vectors.  It is seen that the proposed algorithm with $M_{\mathcal{A}}=16$  surpasses  "PDNet+CENet, fixed $M_{\mathcal{A}}=16$" in high SNR regime, which indicates that different mixed-ADCs allocation will affect the performance notably.  This results show the effectiveness of the proposed  architecture with SELNet to optimize the allocation of mixed-ADCs.
\begin{figure}[!tbp]
\centering
\includegraphics[width=0.4\textwidth]{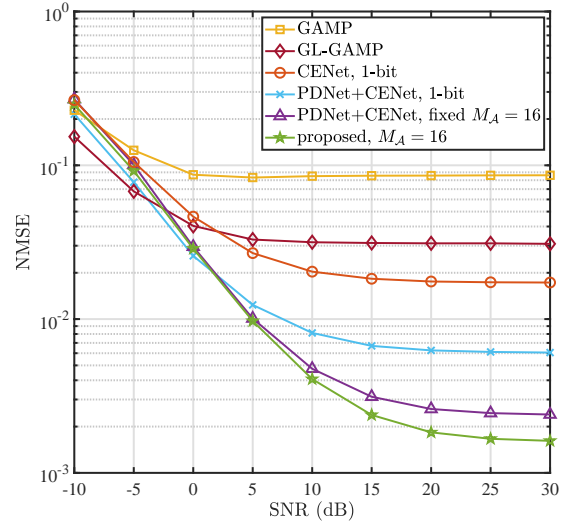}
\caption{$\mathrm{NMSE}$ versus SNR; $M=64$ and $N_{p}=64$. Performance comparison of the proposed method and benchmark algorithms.}
\label{fig:NMSE16}
\end{figure}

\begin{figure}[!tbp]
\centering
\includegraphics[width=0.4\textwidth]{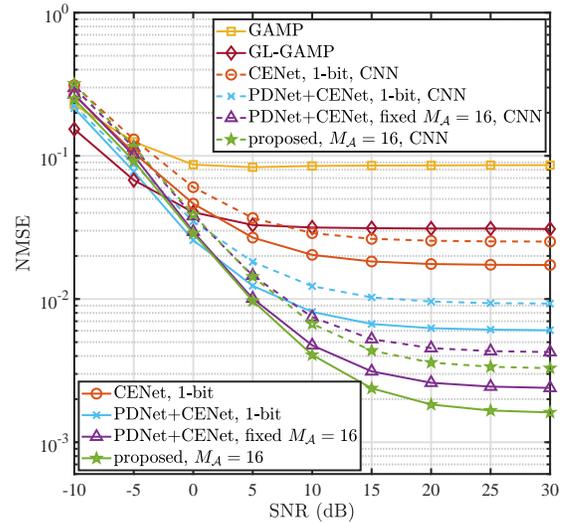}
\caption{$\mathrm{NMSE}$ versus SNR; $M=64$ and $N_{p}=64$. Solid line and dashed line denote the performances of  the densely connected RK3 model-driven blocks and regular CNNs, respectively.}
\label{fig:NMSECNN}
\end{figure}

Fig.~\ref{fig:NMSECNN} shows the NMSE of  channel estimation  algorithms versus SNR, where $M=64$, $N_{p}=64$ and the number of full-resolution ADCs is $M_{\mathcal{A}}=16$ for mixed-ADCs cases. In Fig. \ref{fig:NMSECNN}, the dashed lines denote the results by replacing  the densely connected RK3 model-driven blocks in CENet with regular CNNs. Besides the similar phenomena shown in Fig.~\ref{fig:NMSE16}, we observe that there is a notable gap between the solid line and dashed line for each method. In particular, compared to regular CNNs, there is a significant performance improvement by using the RK3 model-driven blocks with dense connections in CENet. These results show the effectiveness and superior of the proposed  CENet over the regular neural network structures.

\begin{figure}[!tbp]
\centering
\includegraphics[width=0.4\textwidth]{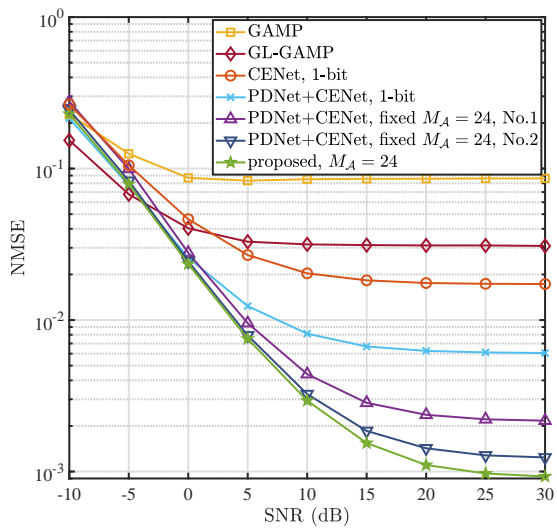}
\caption{$\mathrm{NMSE}$ versus SNR; $M=64$ and $N_{p}=64$. We set $M_{\mathcal{A}}=24$ for the mixed-ADCs cases.}
\label{fig:NMSE24}
\end{figure}

Fig.~\ref{fig:NMSE24} shows the NMSE of  channel estimation  algorithms versus SNR, where $M=64$ and $N_{p}=64$. 
Similar phenomena of Fig. \ref{fig:NMSE16} can also be seen in Fig.~\ref{fig:NMSE24}, where "CENet, 1-bit" is much better than traditional GL-GAMP algorithm and "PDNet+CENet, 1-bit" has significant performance improvement over "CENet, 1-bit".
Additionally, in Fig.~\ref{fig:NMSE24}, we set $M_{\mathcal{A}}=24$ with $24$ full-resolution and $40$ one-bit ADCs utilized. Unlike the configuration of Fig.~\ref{fig:NMSE16}, since $24$ full-resolution ADCs can not be equispaced within $M=64$ antennas, we consider two configurations for mixed-ADCs allocation denoted by "fixed $M_{\mathcal{A}}=24$, No.1" and "fixed $M_{\mathcal{A}}=24$, No.2", respectively. As can be seen, there is a notable performance gap between "PDNet+CENet, fixed $M_{\mathcal{A}}=24$, No.1" and "PDNet+CENet, fixed $M_{\mathcal{A}}=24$, No.2", which indicates that different allocations of full-resolution ADCs affect channel estimation accuracy significantly. Hence,  to illustrate the performance gains by optimizing the allocation of mixed-ADCs, the performance of the proposed  end-to-end architecture with $M_{\mathcal{A}}=24$ is also shown in Fig.~\ref{fig:NMSE24}. We see that the proposed end-to-end architecture outperforms "PDNet+CENet, fixed $M_{\mathcal{A}}=24$, No.1" and "PDNet+CENet, fixed $M_{\mathcal{A}}=24$, No.2" remarkably, which shows the effectiveness of the proposed end-to-end architecture with SELNet to optimize the allocation of mixed-ADCs and improve channel estimation performance.
\begin{figure}[!tbp]
\centering
\includegraphics[width=0.4\textwidth]{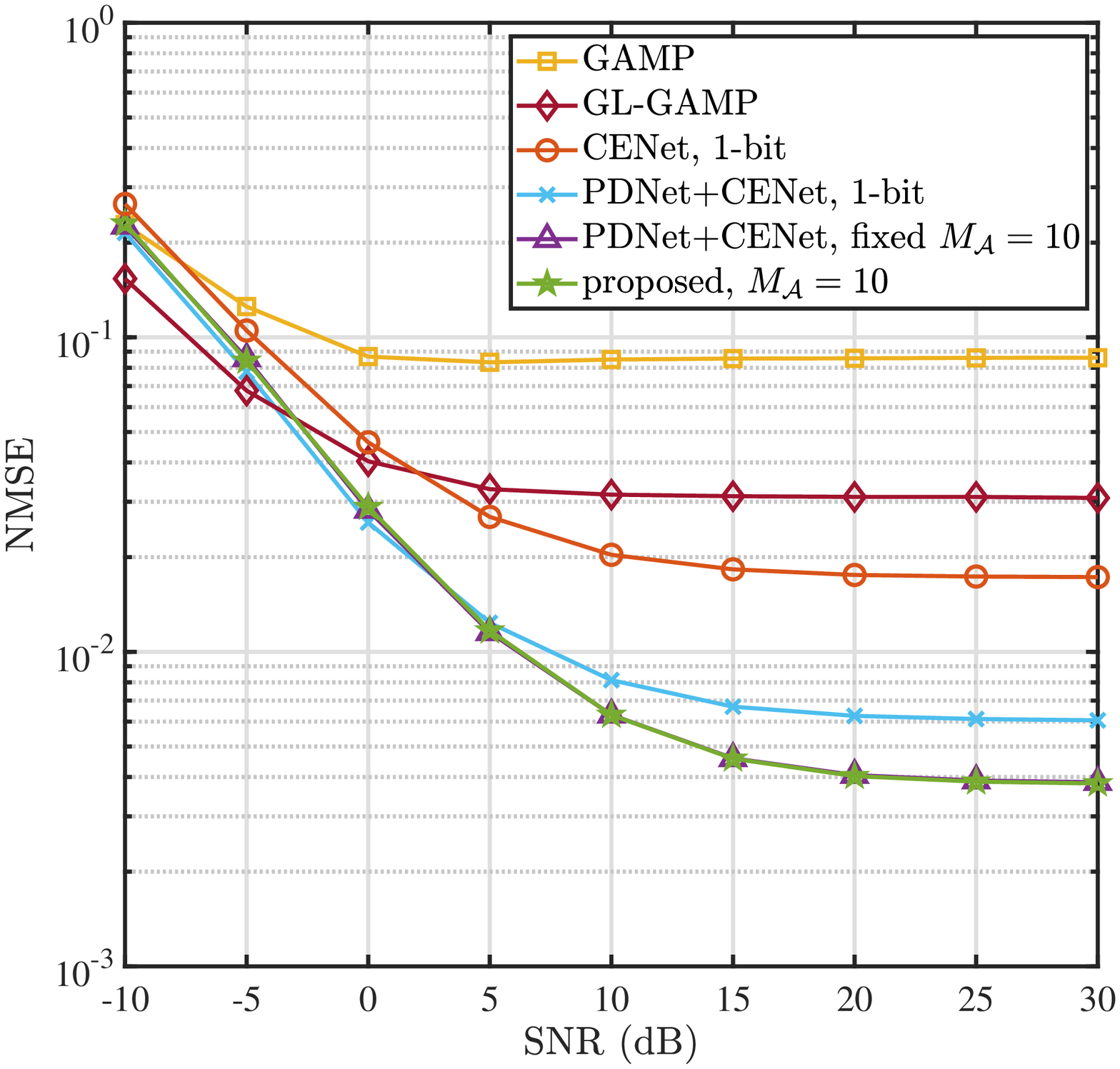}
\caption{$\mathrm{NMSE}$ versus SNR; $M=64$ and $N_{p}=64$. Performance comparison of the proposed method and benchmark algorithms, where $M_{\mathcal{A}}=10$.}
\label{fig:NMSE10}
\end{figure}

Fig.~\ref{fig:NMSE10} displays the NMSE of  channel estimation  algorithms over SNR, where $M=64$ and $N_{p}=64$. 
Similar phenomena for GL-GAMP, "CENet, 1-bit" and "PDNet+CENet, 1-bit" can be seen in Fig.~\ref{fig:NMSE10}.
Moreover, the results for the mixed-ADCs cases are also illustrated in Fig.~\ref{fig:NMSE10}, where $10$ full-resolution and $56$ one-bit ADCs utilized. 
Specifically, for "PDNet+CENet, fixed $M_{\mathcal{A}}=10$", the index set of RF chains in which the full-resolution ADCs are deployed is $\{1,9,17,25,33,41,49,57,13,53\}$. While, for the proposed end-to-end architecture, the allocation for mixed-ADCs is optimized by SELNet. Unlike the results of Fig.~\ref{fig:NMSE16} and Fig.~\ref{fig:NMSE24}, the  curve of the proposed end-to-end method almost overlaps that of "PDNet+CENet, fixed $M_{\mathcal{A}}=10$" as shown in Fig.~\ref{fig:NMSE10}. This indicates that the performance gains by optimizing the mixed-ADCs allocation become smaller when the number of full-resolution ADCs decreases.
\begin{figure}[!tbp]
\centering
\includegraphics[width=0.4\textwidth]{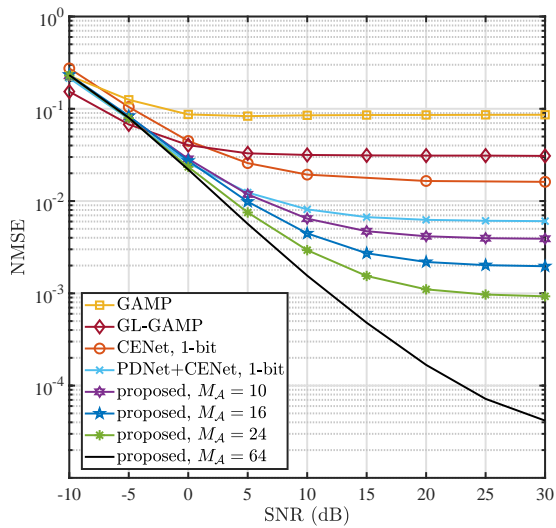}
\caption{$\mathrm{NMSE}$ versus SNR for different values of $M_{\mathcal{A}}$, where $M=64$ and $N_{p}=64$.}
\label{fig:NMSEfull}
\end{figure}

Fig.~\ref{fig:NMSEfull} plots the NMSE of  channel estimation  algorithms over SNR by varying the value of $M_{\mathcal{A}}$, where $M=64$ and $N_{p}=64$.
As can been seen in Fig.~\ref{fig:NMSEfull}, the channel estimation error of the proposed end-to-end architecture decreases when the number of full-resolution ADCs, $M_{\mathcal{A}}$, increases. The reason is straightforward that less quantization errors will be introduced when more full-resolution ADCs are utilized. Meanwhile,  there are nonzero error floors for all the methods with mixed-precision ADCs ($M_{\mathcal{A}}<64$), especially in the high SNR regime. Nevertheless, we observe that the nonzero error floor of the proposed end-to-end method is eliminated when all RF chains are equipped with full-resolution ADCS ($M_{\mathcal{A}}=64$). This indicates that the nonzero error floor is caused by the amplitude information loss introduced by one-bit quantization in the high SNR regime, which is also called stochastic resonance phenomenon. Moreover, these results show that one way  to increase the estimation accuracy is to increase the ratio of high-resolution ADCs in mixed-precision architecture.

\begin{figure}[!tbp]
\centering
\includegraphics[width=0.4\textwidth]{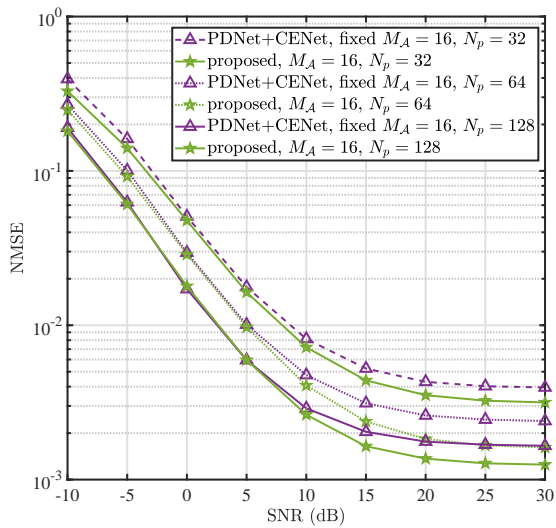}
\caption{$\mathrm{NMSE}$ versus SNR for different values of $N_{p}$, where $M=64$ and $M_{\mathcal{A}}=16$.}
\label{fig:NMSEpilot}
\end{figure}
Next, we study the impact of the pilot length $N_{p}$ on the channel estimation performance. 
Fig.~\ref{fig:NMSEpilot} displays the NMSE of the proposed method and "PDNet+CENet, fixed $M_{\mathcal{A}}=16$", where $M=64$ and the $16$ full-resolution ADCs are equispaced within the $M=64$ antennas for the "PDNet+CENet, fixed $M_{\mathcal{A}}=16$". It can been seen that the proposed end-to-end method outperforms "PDNet+CENet, fixed $M_{\mathcal{A}}=16$" for each value of $N_{p}$, which indicates that the proposed end-to-end architecture with SELNet can optimize the allocation of mixed-ADCs to further improve channel estimation accuracy. 
Moreover, the results also show that increasing the pilot length $N_{p}$ can improve the performances of both the methods. The reason is that increasing $N_{p}$ will increase the transmission energy for fixed  average transmission power $\rho$. 
Furthermore, we see that the proposed method for $N_{p}=64$ has the same performance with "PDNet+CENet, fixed $M_{\mathcal{A}}=16$, $N_{p}=128$" in the high SNR regime, which indicates that the pilot overhead can be reduced to half with optimized mixed-ADCs allocation. 

\begin{figure}[!tbp]
\centering
\includegraphics[width=0.4\textwidth]{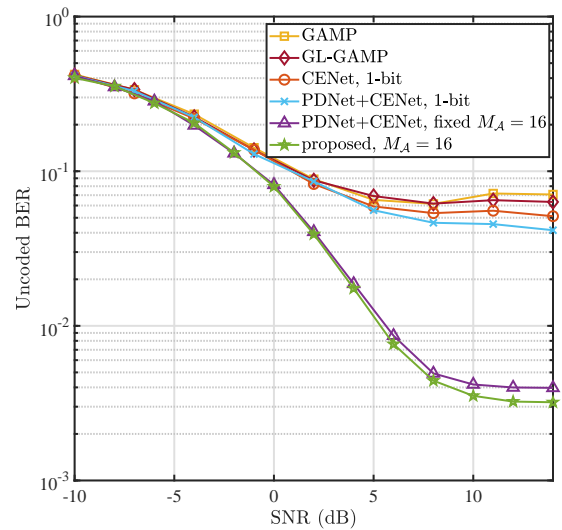}
\caption{BER versus SNR with different channel estimation results, where $M=64$, $M_{\mathcal{A}}=16$, and 16QAM is used.
\label{fig:BER16QAM}}
\end{figure}

For the last numerical study, we show the effects of different channel estimation methods on the data detection. Specifically, during payload data transmission, the received signal is 
\begin{equation}
\mathbf{y} = \mathcal{Q}\big(\mathbf{H} \mathbf{x} + \mathbf{w}\big),
\end{equation}
where $\mathbf{x}\in {\mathbb C}^{K}$, $\Vert \mathbf{x} \Vert^{2}_{2} = \rho$, and $\mathbf{w}\in {\mathbb C}^{M}$ is Gaussian noise with zero mean and variance $\sigma^2$. Similar to \eqref{eq:RxVectorReal}, the real-valued form of the received signal is given by
$
\widetilde{\mathbf{y}} = \mathcal{Q}\big(\widetilde{\mathbf{H}} \widetilde{\mathbf{x}} + \widetilde{\mathbf{w}} \big).
$
The maximum likelihood (ML) detector for $\widetilde{\mathbf{x}}$ is formulated as 
\begin{equation}\label{eq:ML}
\argmax_{\hat{\mathbf{x}}\in\mathcal{S}} \prod_{i=1}^{\widetilde{\mathcal{B}}} \Phi\left( \sqrt{\frac{2}{\sigma^2}} \widetilde{y}_{i} \widetilde{\mathbf{h}}_{i}^{T} \hat{\mathbf{x}} \right)   \prod_{j=1}^{\widetilde{\mathcal{A}}} \frac{1}{\sqrt{\pi} \sigma } e^{-\frac{(\widetilde{\mathbf{h}}_{j}^{T} \hat{\mathbf{x}} - \widetilde{y}_{j} )^{2}}{\sigma^2}},
\end{equation}
where $\mathcal{S}$ is the constellation set, $\widetilde{\mathbf{h}}_{i}^{T}$ is the $i$th row of $\widetilde{\mathbf{H}}$, $\widetilde{\mathcal{A}}=\{\mathcal{A},\mathcal{A}+M\}$ and $\widetilde{\mathcal{B}}=\{\mathcal{B},\mathcal{B}+M\}$. The exhaustive search over $\mathcal{S}$ is required to solve \eqref{eq:ML}, which is computationally intensive. Inspired by \cite{nML}, we relax the constraint $\hat{\mathbf{x}}\in\mathcal{S}$  as
\begin{equation}\label{eq:nML}
\argmax_{\hat{\mathbf{x}}\in {\mathbb R}^{2K} \atop \Vert \hat{\mathbf{x}} \Vert^{2}_{2} \leq \rho} \prod_{i=1}^{\widetilde{\mathcal{B}}} \Phi\left( \sqrt{\frac{2}{\sigma^2}} \widetilde{y}_{i} \widetilde{\mathbf{h}}_{i}^{T} \hat{\mathbf{x}} \right)   \prod_{j=1}^{\widetilde{\mathcal{A}}} \frac{1}{\sqrt{\pi} \sigma } e^{-\frac{(\widetilde{\mathbf{h}}_{j}^{T} \hat{\mathbf{x}} - \widetilde{y}_{j} )^{2}}{\sigma^2}}.
\end{equation}
Note that \eqref{eq:nML} is a convex optimization problem, where the objective is log-concave and the constraint is convex. Hence, we use the aforementioned channel estimation methods to obtain $\mathbf{H}$, and then apply CVX toolbox to solve \eqref{eq:nML}, where SNR (dB) is defined as $10\log_{10}\left( \frac{\rho}{\sigma^2} \right)$. 

Fig.~\ref{fig:BER16QAM} plots the uncoded bit error rate (BER) of the detector in \eqref{eq:nML} with different channel estimation results, where $M=64$, $M_{\mathcal{A}}=16$, and 16QAM is used. 
Comparing Fig.~\ref{fig:BER16QAM} with Fig.~\ref{fig:NMSE16}, we find that higher channel estimation accuracy leads to lower BER. 
Moreover, there is a remarkable performance gap between one-bit and mixed-precision channel estimation, which shows the superiority of the mixed-ADCs architecture. 
Furthermore, it can be seen that the proposed end-to-end estimation method outperforms "PDNet+CENet, fixed $M_{\mathcal{A}}=16$", which shows the effectiveness of the mixed-ADCs allocation optimization. 
Meanwhile, there are nonzero error floors for all methods in the high SNR regime, since the amplitude information is lost during one-bit quantization.

\section{Conclusions} \label{sec:con}
In this paper, by applying DL methods, we investigate the joint pilot design and channel estimation as well as  mixed-ADCs allocation problem for mmWave massive MIMO. 
For the channel estimator, we proposed a Runge-Kutta model-driven densely connected network that can alleviate the vanishing-gradient problem. 
We devise a pilot design network where the optimized pilots are obtained directly from the trained weights. 
Moreover, for the mixed-ADCs allocation optimization, we develop a selection network to choose the antennas for mixed-ADCs allocation.  
Furthermore, we adopt an autoencoder-based end-to-end architecture to jointly train these networks.
Numerical results have been carried out to show the superior performance of the proposed methods in channel estimation.

% \appendices{}
% \section{proof of Theorem \ref{theorem:K-hot}}\label{app:1}

\bibliographystyle{IEEEtran}
\bibliography{IEEEabrv,ADCpilot}

\end{document}